\documentclass[11pt]{article} 
\usepackage{tda2}

\title{\vspace{-5ex}Catching Polygons}
\date{\vspace{-5ex}}
\author{Bradley McCoy \thanks{School of Computing, Montana State U.
} \and Eli Quist \thanks{Dept. of Mathematical
    Sciences, Montana State U.     \newline {\scriptsize {\tt bradleymccoy@montana.edu, eli.quist@student.montana.edu,
    annaschenfisch@montana.edu} }} \and
    Anna Schenfisch\footnotemark[2] }

\begin{document}
\thispagestyle{empty}
\maketitle

\begin{abstract}
Consider an arrangement of $k$ lines intersecting the unit square.  There is
some minimum scaling factor so that any placement of a rectangle with aspect
ratio $1 \times p$ with $p\geq 1$ must non-transversely intersect some portion
of the arrangement or unit square. Assuming the lines of the arrangement are
axis-aligned, we show the optimal arrangement depends on the aspect ratio of the
rectangle. In particular, the optimal arrangement is either evenly spaced parallel
lines or an evenly spaced grid of lines.  We present the precise aspect ratios
of rectangles for which each of the two nets are optimal.
\end{abstract}

\section{Introduction}
\label{sec:intro}

During the open problem session of CCCG20, Joseph O'Rourke suggested the
following problem.  Consider an arrangement of $k$ lines, all of which intersect
the unit square. For a fixed polygon $P$, there is some minimum scale factor $c
> 0$ such that the scaled polygon $cP$ cannot be embedded in the unit square
without intersecting any of the $k$ lines of the arrangement non-transversely
(allowing translation and rotation).
That is, the scaled polygon will `just touch' at least one line of the
arrangement.  Can we compute the minimum such $c$ over all possible arrangements
of this type?  Can we describe an arrangement that realizes this minimum?  See
\figref{intro} for an example. 

This problem can be described using an analogy to tripwire lasers.
In this analogy, the polygon is an intruder in the unit square.
The intruder can vary in size but always has the same shape.
The goal is to minimize the size of the intruder that avoids the net of lasers.

\begin{figure}[htb]
    \begin{minipage}[l]{0.67\textwidth} 
        \includegraphics[scale=0.7]{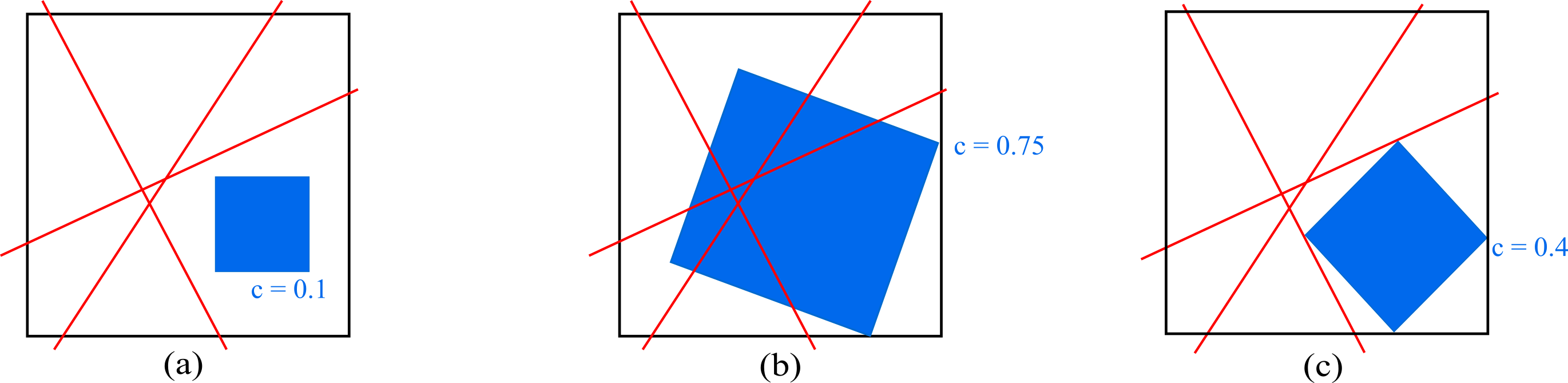}
    \end{minipage}\hfill
    \begin{minipage}[c]{0.3\textwidth} 
	\caption{(a) If $c$ is small, the polygon can avoid the lasers.
        (b) If $c$ is large, any rotation and translation of the polygon is
        caught in the lasers.
        (c) The minimum $c$ such that the polygon is caught in
        the lasers. }
			\label{fig:intro}
    \end{minipage}
\end{figure}

Here, we consider the special case where the lines are
axis-aligned and $P$ is a rectangle (a lemma towards justifying the setting of
axis-aligned lines is given in \appendref{append-axis}). 
We observe that, depending on the aspect ratio
of the rectangle being considered, the optimal arrangement is either evenly
spaced parallel lines or a grid of lines. See \figref{intro-2} for an example.
O'Rourke asked, \emph{for what aspect ratios of rectangles is each of the two nets
optimal?}
In this work, we answer this question precisely.
We are only aware of one work that directly considers this problem \cite{aghamolaei-capfwamn-20}, 
however the author has a 
different interpretation of the problem.
Similar laser based
localization problems are considered in \cite{ArkinD0GMPT20,boundary-lasers}.
Using the results from \cite{ADS-95},
given any net and aspect ratio for the intruder one can compute 
the optimal scale factor.

\begin{figure}[htb]
    \begin{minipage}[l]{0.55\textwidth} 
        \includegraphics[scale=.5]{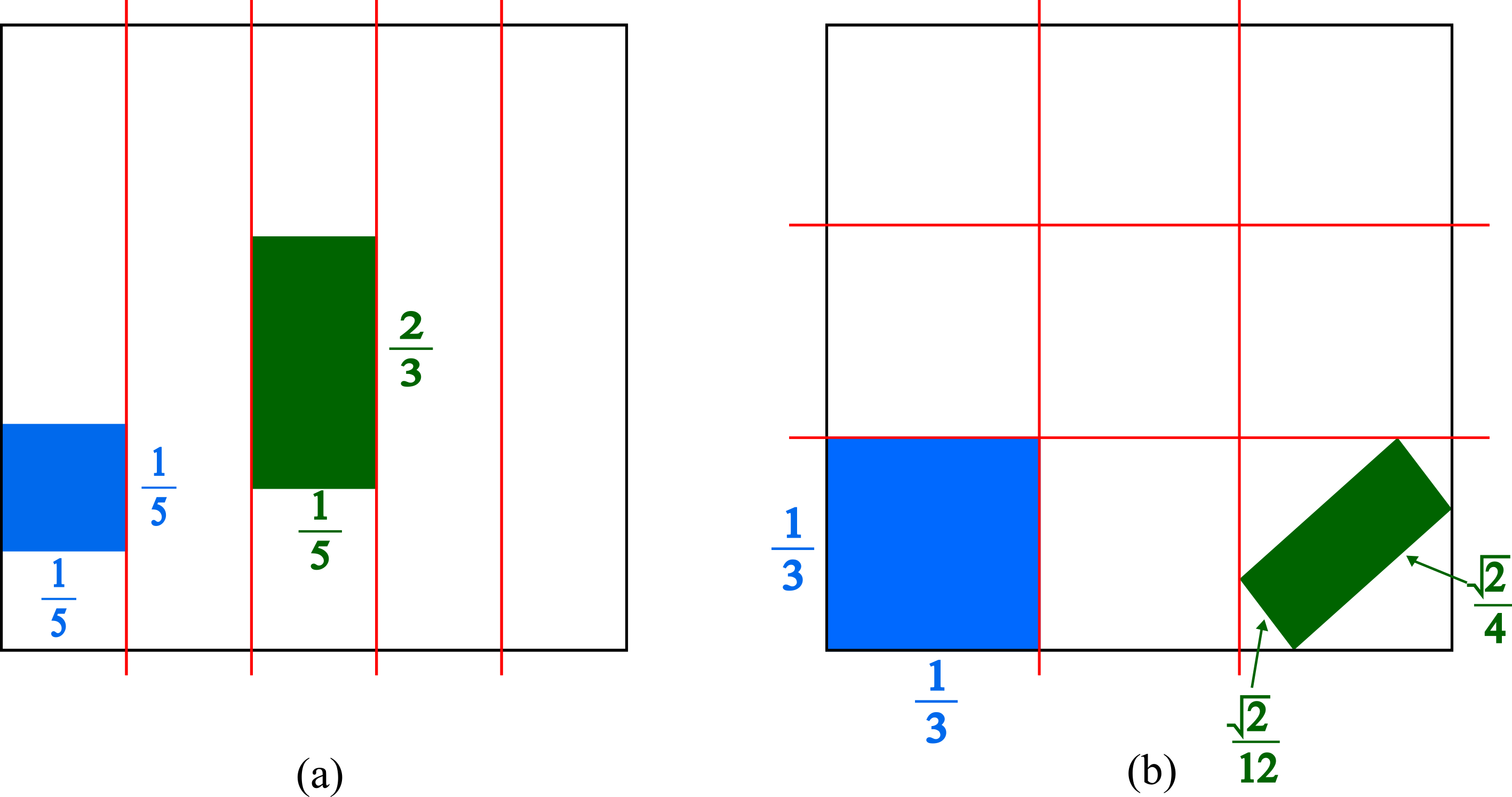}
    \end{minipage}
    \begin{minipage}{0.45\textwidth}
        \caption{(a) For a square the evenly spaced vertical lines are optimal
        but for a rectangle with a $(1\times 3)$ aspect ratio evenly spaced
        vertical lines is not optimal.
	(b) For a square the evenly spaced grid is not optimal but for
	a $(1\times 3)$ rectangle the evenly spaced grid is optimal.}
		\label{fig:intro-2}
    \end{minipage}
\end{figure}

\section{Rectangular Intruders in Rectangular Nets}

\label{sec:curve}

When $P$ is a rectangle and the $k$ lines are axis-aligned,
the problem of calculating the optimal scale factor $c$
reduces to finding the largest rectangle inscribed inside of another
rectangle. Inscribing rectangles inside rectangles has
been studied in \cite{Carver-1957,Dunkel-1920,wetzel-00}.
In this section, we construct a curve that describes
the optimal scale factor $c$ for rectangular intruders in rectangular nets for
all aspect ratios of the inscribed rectangle.

We fix the aspect ratio of the hole in the net
to be $1\times n$ with $n\geq 1.$
Let the aspect ratio of the intruder be $1\times p$ with $p\geq 1.$
We express the optimal scale factor 
$c$ as a function of the aspect ratio
$p$. We denote this curve by $\curve$. 
See \figref{curve} for an example.

\begin{figure}[htb]
        \centering
    
        \includegraphics[width=.4\textwidth]{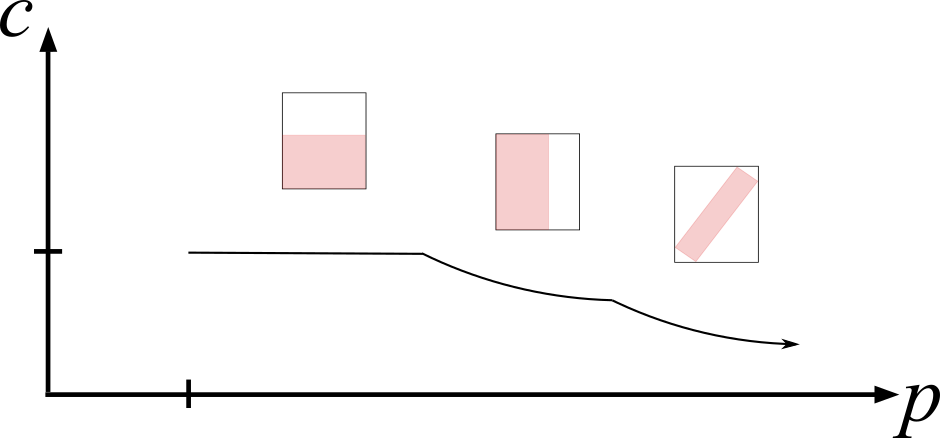}

		\caption{The inscribing curve for some $n \in \R^{\geq 1}$.
		The $x$-axis is the aspect ration of the intruder and the $y$-axis
		is the optimal scale factor.}
			\label{fig:curve}
\end{figure}

The curve $\curve$ consists of three parts.
For small $p\leq n$, placing the shorter side of the intruder
parallel to the shorter side of the net is optimal and
$c=1.$
For medium size $p$ values relative to $n$,
the value of $c$ is limited by the height of the net and the 
longer side of the intruder is scaled to equal the longer side of the net.
This gives $pc=n$ or $c=\frac{n}{p}$.

\begin{wrapfigure}{l}{0.5\textwidth}
    \centering
    \raisebox{0pt}[\dimexpr\height-1\baselineskip\relax]
    {\includegraphics[width=.25\textwidth]{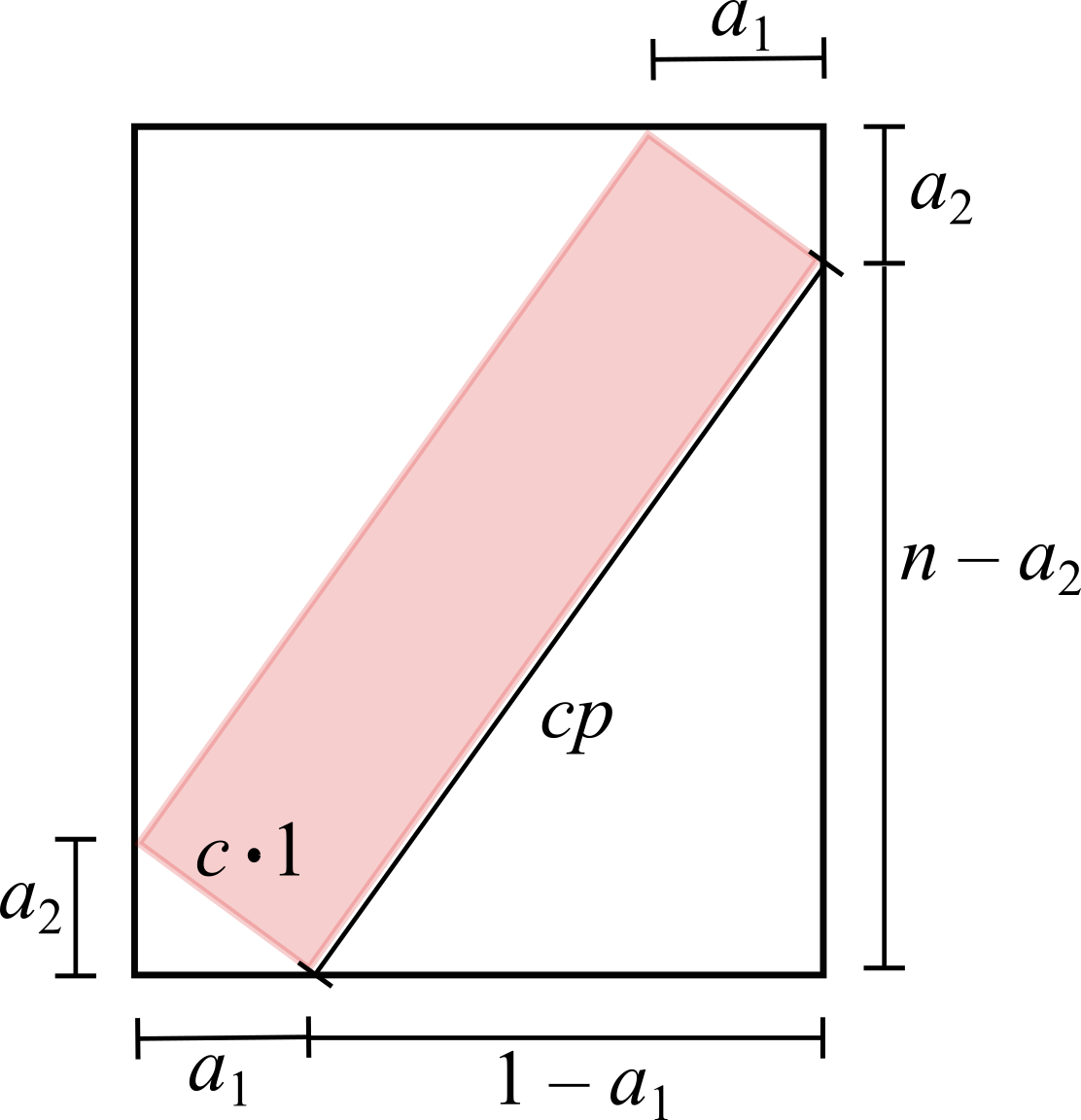}}
                \caption{A diagonally inscribed rectangle.\label{fig:diagfish}}
			
                        \vspace{-1.7cm}
\end{wrapfigure}
For large $p$ relative to $n$, placing the intruder diagonally is optimal.
Using the variables indicated in \figref{diagfish}, we have
the following three equations

\begin{align}
\centering
    &\frac{a_1}{a_2}=\frac{n-a_2}{1-a_1}\label{eqn:similar}\\
       &a_1^2+a_2^2=c^2\label{eqn:right}\\
      &(1-a_1)^2+(n-a_2)^2=(cp)^2\label{eqn:ceepee}
\end{align}
along with the natural constraints of the problem, $0<c,\ 0<a_1<1,\ 0<a_2<n$. 
\eqnref{similar} is due to all triangles being similar. \eqnref{right} and \eqnref{ceepee} are
applications of Pythagorean's~theorem.

Solving for $c$ gives
$$c=\sqrt{\left(\frac{n-p}{1-p^2}\right)^2(p^2+1)-2\left(\frac{n-p}{1-p^2}\right)p+1}.$$

Therefore, we are able to give an explicit formula for $\curve$
for a given $n$, namely,
\[ \curve = \begin{cases} 
      1 & 1\leq p\leq n \\
      \frac{n}{p} & n<p\leq w_n \\
      \sqrt{\left(\frac{n-p}{1-p^2}\right)^2(p^2+1)-2\left(\frac{n-p}{1-p^2}\right)p+1} & w_n<p
   \end{cases}
\]
where $w_n$ is the solution to 
$\sqrt{\left(\frac{n-p}{1-p^2}\right)^2(p^2+1)-2\left(\frac{n-p}{1-p^2}\right)p+1}=\frac{n}{p}$ for $1<n<p.$
The value $w_n$ represents the scale factor where
the vertically inscribed rectangle has the same scale factor as the diagonally inscribed rectangle.

When $k$ is even, the grid with $\frac{k}{2}$ horizontal lines and $\frac{k}{2}$ 
vertical lines  has 
square holes with side length $\frac{1}{\frac{k}{2}+1}$.
The arrangement with $k$ vertical and $0$ horizontal lines has 
rectangular holes with dimension $1\times \frac{1}{k+1}$.
The curves $\frac{1}{\frac{k}{2}+1} \mathcal{C}_1(p)$ and
$\frac{1}{k+1}\mathcal{C}_{k+1}(p)$ intersect at $p=\frac{k+1}{\frac{k}{2}+1}.$
 For even values of $k$,
we define the \emph{base curve} to be

\[ B_k(p) =\min\bigg\{\frac{1}{k+1}\mathcal{C}_{k+1}(p), \frac{1}{\frac{k}{2}+1} \mathcal{C}_1(p) \bigg\}
\]

The case for odd values of $k$ is similar and included in \appendref{odd}.

    


\section{Optimality Results}
\label{sec:optimal}
In this section, we show that for any axis-aligned net $\mathcal{N}$
with an even number of lines
and any aspect ratio of the intruder $p$, the base
curve has a scale factor that is less than or equal 
to the scale factor of $\mathcal{N}$.
Let $\mathcal{N}(k,0)$ denote the net with $k$ evenly spaced
parallel lines and let $\mathcal{N}(\frac{k}{2},\frac{k}{2})$ denote the
net with $\frac{k}{2}$ evenly spaced vertical lines and $\frac{k}{2}$
evenly spaced horizontal lines.
The rectangle aspect ratio where the base curve
switches from $\mathcal{N}(k,0)$ to $\mathcal{N}(\frac{k}{2},\frac{k}{2})$
is $p=\frac{k+1}{\frac{k}{2}+1}$. We now state our
main theorem:

\begin{theorem}[Regular is Optimal]\label{thm:optimal}
For $k$ even, the optimal axis aligned net for rectangular polygons is
$\mathcal{N}(k,0)$ for aspect ratio $p\leq\frac{k+1}{\frac{k}{2}+1}$
and $\mathcal{N}(\frac{k}{2},\frac{k}{2})$ for aspect ratio $p\geq\frac{k+1}{\frac{k}{2}+1}$.
\end{theorem}

\begin{proof}
First, notice that regularly spaced lines give a smaller scale factor than irregularly
 spaced lines.
This is because a rectangular hole generated by regular spacing
has height and width equal to the average. 
A rectangular hole generated by irregular spacing has
a hole with height and width greater than or equal to the average.

Consider any axis aligned net over the square,
let $v$ be the number of vertical lines and $h$ be the number
of horizontal lines call the net $\mathcal{N}(v,h)$. Notice $v+h=k$.
If $v=n$ then we have a $\frac{k}{2}\times \frac{k}{2}$ net
and if $n=0$ (or $v=0$), then we have $k$ parallel lines.
The base curve is defined to be the minimum scale factor of
 $\mathcal{N}(\frac{k}{2},\frac{k}{2})$ and $\mathcal{N}(k,0)$.
Thus, the scale factor of the base curve is
less than or equal to either of these nets.

Consider any other $v$ and  $h$.
Without loss of generality,
assume $h<v$ we have $0<h<\frac{k}{2}<v<k.$
There are $(v+1)(h+1)$ holes in the net and 
the average size of a hole is $\frac{1}{v+1}\times \frac{1}{h+1}$.
There exists one hole at least as big as the average,
so, we have a hole at least as big as the hole with aspect ratio 
$\frac{v+1}{h+1}$ scaled so the smaller side has 
length $\frac{1}{v+1}$. 
The optimal scale factor of a rectangle with aspect ratio $p$ that fits
inside this hole is given by
$\frac{1}{v+1}\mathcal{C}_{\frac{v+1}{h+1}}(p)$.
We compare $\frac{1}{v+1}\mathcal{C}_{\frac{v+1}{h+1}}(p)$
to $B_k(p)$. 
Both $\frac{1}{v+1}\mathcal{C}_{\frac{v+1}{h+1}}(p)$
and $B_k(p)$ are piecewise functions,
we directly examine all $p\geq 1$ to show
$B_k(p)\leq\frac{1}{v+1}\mathcal{C}_{\frac{v+1}{h+1}}(p)$, see
\figref{optimal} for intuition.

For $1\leq p\leq \frac{k+1}{\frac{k}{2}+1}$,
we have 
$B_k(p)=\frac{1}{k+1}$
and $\frac{1}{v+1}\mathcal{C}_{\frac{v+1}{h+1}}(p)=\frac{1}{v+1}.$ 
since $v<k$, we have $\frac{1}{k+1}<\frac{1}{v+1}$.
For $\frac{k+1}{\frac{k}{2}+1}\leq p\leq \frac{v+1}{h+1}$,  $B_k(p)$ 
decreases and $\frac{1}{v+1}\mathcal{C}_{\frac{v+1}{h+1}}(p)=\frac{1}{k+1}$ is constant.

For $\frac{v+1}{h+1}<p\leq w_{\frac{v+1}{h+1}}$, 
$B_k(p) =\min\bigg\{\frac{1}{k+1}\mathcal{C}_{k+1}(p), \frac{1}{\frac{k}{2}+1} \mathcal{C}_1(p) \bigg\}
\leq
\left(\frac{1}{\frac{k}{2}+1}\right)\left(\frac{1}{p}\right)$
and~$\frac{1}{v+1}\mathcal{C}_{\frac{v+1}{h+1}}(p)=\frac{1}{v+1}\left(\frac{v+1}{h+1}\right)\frac{1}{p}.$
Since $h\leq \frac{k}{2}$,
we have 
$$\left(\frac{1}{\frac{k}{2}+1}\right)\left(\frac{1}{p}\right)<\frac{1}{h+1}\left(\frac{1}{p}\right)=\frac{1}{v+1}\left(\frac{v+1}{h+1}\right)\frac{1}{p}.$$

For $p\geq w_{\frac{v+1}{h+1}}$ the interior rectangle is placed
diagonally. 
Let $c'$ be the scale factor value of the rectangle placed diagonally in $\mathcal{N}(\frac{k}{2},\frac{k}{2}).$
Consider the length of the rectangle, with shorter side equal to $c'$,
 placed diagonally in a rectangle with
sides $\frac{1}{v+1}\times\frac{1}{h+1}.$
Let $a_1$ and $a_2$ be the legs of the small
right triangle formed by $c'$,
the squared length of this inscribed rectangle
is $\ell^2(v,h,a_1,a_2)=(\frac{1}{v+1}-a_2)^2+(\frac{1}{h+1}-a_1)^2$.
The minimum of this function along the constraints 
\eqnref{similar}, \eqnref{right}, and $v+h=k$
occurs when $v=h$. We omit the details, but this can be done
with Lagrange multipliers.
So, when $v\neq h$ we can fit the diagonal rectangle of 
$\mathcal{N}(\frac{k}{2},\frac{k}{2})$ inside the diagonal of the 
rectangle with dimensions $\frac{1}{v+1}\times\frac{1}{h+1}$
so the optimal scale factor must be~larger.
\end{proof}

\begin{figure}[htb]
        \centering
    
        \includegraphics[width=.36\textwidth]{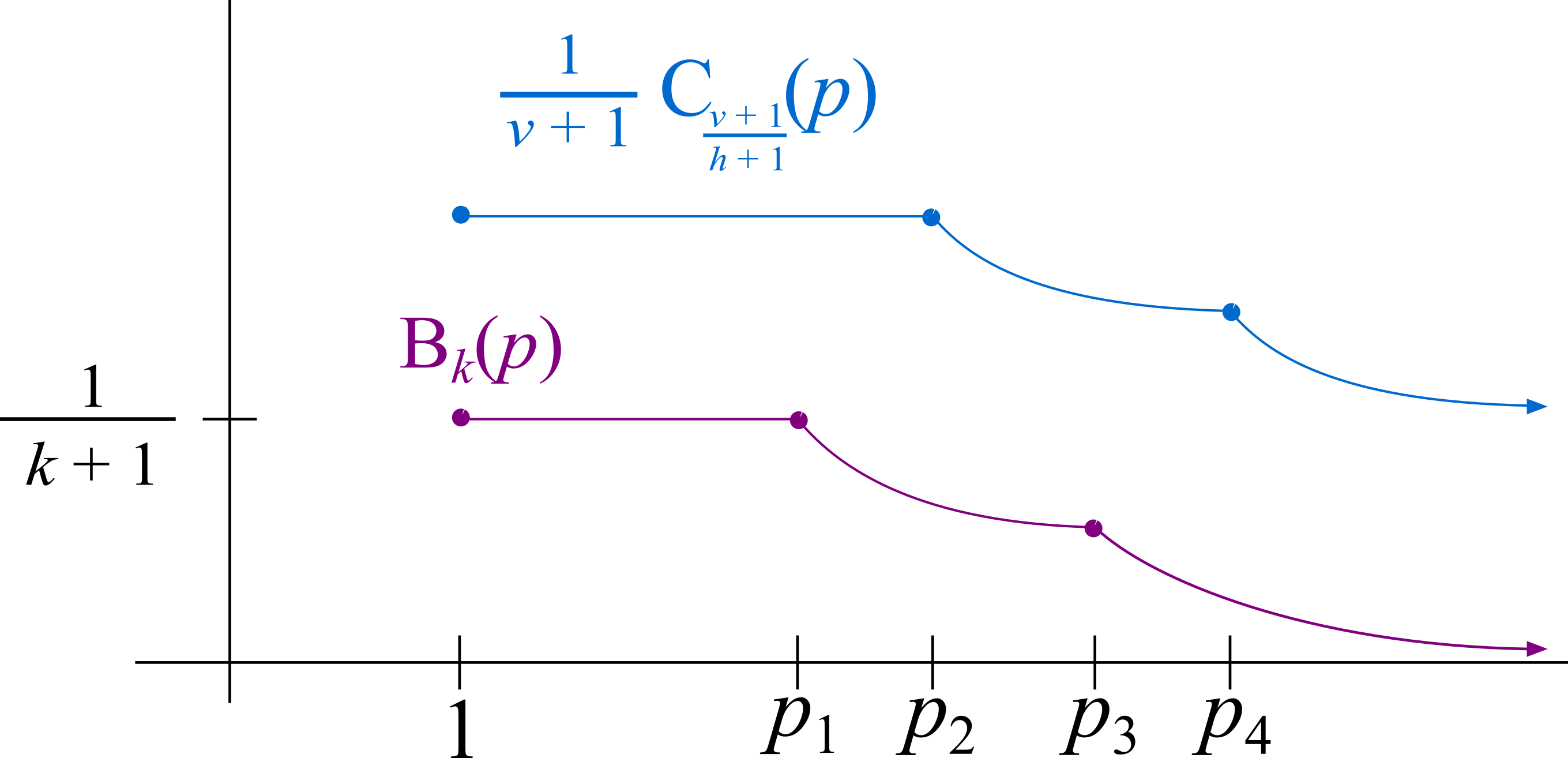}

		\caption{The curves $B_k(p)$ and $ \frac{1}{v+1}\mathcal{C}_{\frac{v+1}{h+1}}(p)$. Here, $p_1 = \frac{k+1}{\frac{k}{2}+1}$, $p_2 = \frac{v+1}{h+1}$, $p_3 = \sqrt{2}+1$, and $p_4 = w_{\frac{v+1}{h+1}}$.}
			\label{fig:optimal}
\end{figure}

\section{Discussion}

\label{sec:discuss}

In this work, we showed that for axis-aligned nets with $k$ lines and rectangular
intruder with aspect ratio $1\times p$, the optimal net is evenly spaced parallel lines for
$p\leq \frac{k+1}{\frac{k}{2}+1}$ and a grid of $\frac{k}{2}\times \frac{k}{2}$ evenly 
spaced lines for $p\geq \frac{k+1}{\frac{k}{2}+1}.$
As far as we know, the problem is still open for non axis-aligned nets and rectangular intruders.
Our hope is that someone can show, that for any net, 
one can construct an axis-aligned
net that does change the optimal scale factor. 
Non-rectangular intruders would also be interesting to explore.

\small{
\bibliographystyle{siam}
\bibliography{biblio}
}
\appendix
\section{An odd number of lines}
\label{append:odd}

In this section, we prove that, when $k$ is odd,
the net $\mathcal{N}(k,0)$ is optimal for $1\leq p\leq \frac{(k+1)\lfloor\frac{k}{2}\rfloor}{\lceil\frac{k}{2} \rceil^2}$
and the net $\mathcal{N}(\lceil\frac{k}{2}\rceil,\lfloor\frac{k}{2}\rfloor)$ is optimal
for $p\geq  \frac{(k+1)\lfloor\frac{k}{2}\rfloor}{\lceil\frac{k}{2} \rceil^2}.$
The net $\mathcal{N}(\lceil\frac{k}{2}\rceil,\lfloor\frac{k}{2}\rfloor)$
has $\frac{1}{\lceil\frac{k}{2}\rceil }\times  \frac{1}{\lfloor\frac{k}{2}\rfloor }$ holes
and the net $\mathcal{N}(k,0)$ has $1\times \frac{1}{k+1}$ holes.
We consider the curves $\frac{1}{\lceil\frac{k}{2}\rceil} \mathcal{C}_{\frac{\lceil\frac{k}{2}\rceil}{\lfloor\frac{k}{2}\rfloor}}(p)$ and
$\frac{1}{k+1}\mathcal{C}_{k+1}(p)$. See \figref{odd}.
These curves 
intersect at $p=\frac{(k+1)\lfloor\frac{k}{2}\rfloor}{\lceil\frac{k}{2} \rceil^2}.$
We define the \emph{base curve} for $k$ odd to be

\[ D_k(p) = \min\Big\{  \frac{1}{k+1}\mathcal{C}_{k+1}(p), \frac{1}{\lceil\frac{k}{2}\rceil} \mathcal{C}_{\frac{\lceil\frac{k}{2}\rceil}{\lfloor\frac{k}{2}\rfloor}}(p)\Big\}.
\]

\begin{figure}[htb]
        \centering
    
        \includegraphics[width=.4\textwidth]{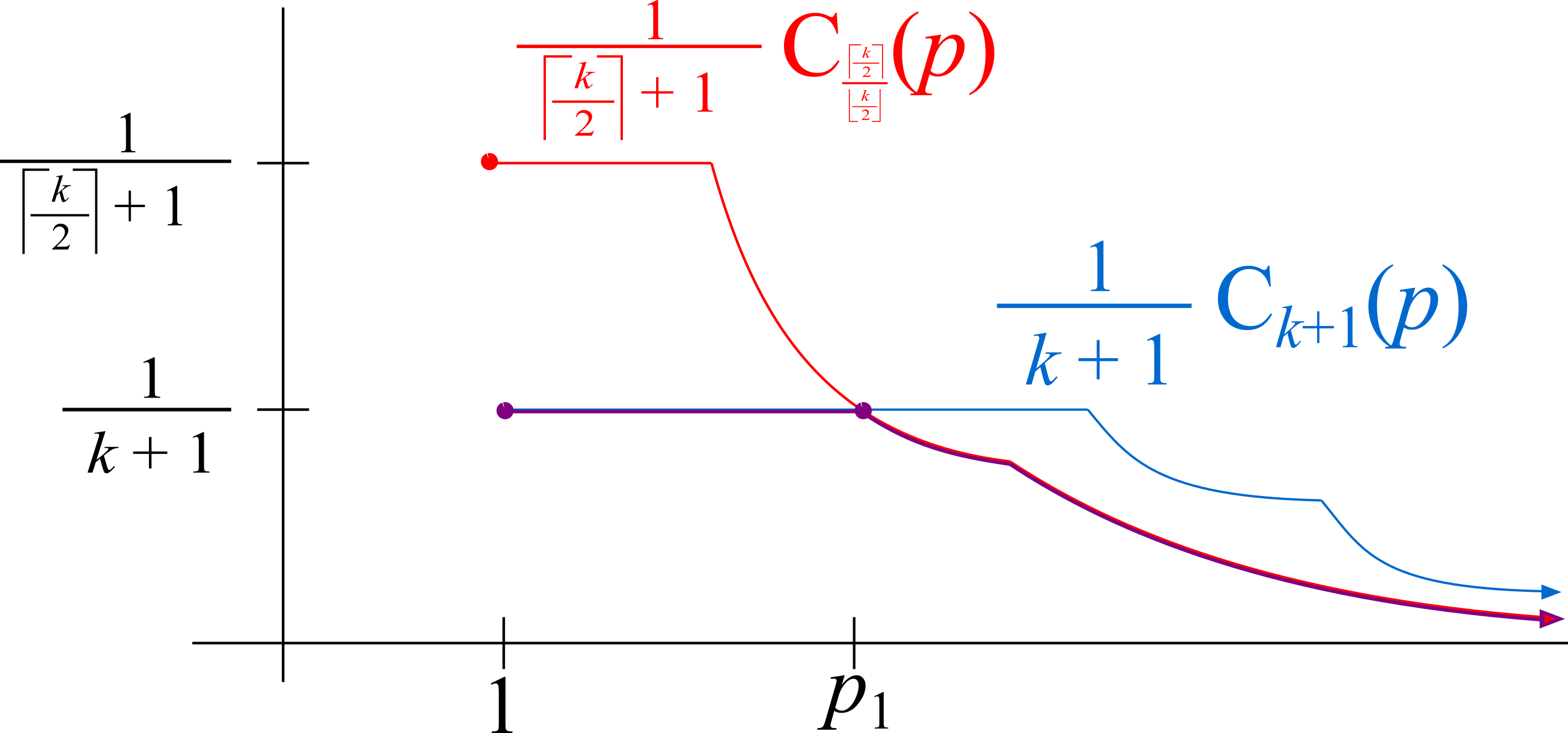}

		\caption{The minimum curve for $k$ odd, with $p_1 = \frac{(k+1)\lfloor\frac{k}{2}\rfloor}{\lceil\frac{k}{2} \rceil^2}$.}
			\label{fig:odd}
\end{figure}

\begin{theorem}[Regular is Optimal Odd]\label{thm:optimal-odd}
For $k$ odd, the optimal axis aligned net for rectangular polygons is
$\mathcal{N}(k,0)$ for aspect ratio $p\leq\frac{(k+1)\lfloor\frac{k}{2}\rfloor}{\lceil\frac{k}{2} \rceil^2}$
and $\mathcal{N}(\lceil\frac{k}{2}\rceil,\lfloor\frac{k}{2}\rfloor)$ for $p\geq\frac{(k+1)\lfloor\frac{k}{2}\rfloor}{\lceil\frac{k}{2} \rceil^2}$.
\end{theorem}

\begin{proof}

Consider any axis aligned net over the square,
let $v$ be the number of vertical lines and $h$ be the number
of horizontal lines call the net $\mathcal{N}(v,h)$. Notice $v+h=k$.

Recall evenly spaced lines give a smaller scale factor than irregularly spaced lines.
If $v=\lceil\frac{k}{2}\rceil$ and $ \lfloor\frac{k}{2}\rfloor=n$ then we have a 
$\lceil\frac{k}{2}\rceil\times \lfloor\frac{k}{2}\rfloor$ net.
If $n=0$ (or $v=0$), then we have $k$ parallel lines.
The base curve is defined to be the minimum scale factor of
$\mathcal{N}(k,0)$ and $\mathcal{N}(\lceil\frac{k}{2}\rceil,\lfloor\frac{k}{2}\rfloor)$
so the base curve has scale factor less than or equal to $\mathcal{N}(k,0).$

Consider any other $v$ and  $h$.
Without loss of generality
assume $h<v$ we have
$0<h< \lfloor\frac{k}{2}\rfloor<\lceil \frac{k}{2} \rceil< v<k$

There are $(v+1)(h+1)$ holes in the net and 
the average size of a hole is $\frac{1}{v+1}\times \frac{1}{h+1}$.
There exits one hole at least as big as the average,
that is, with width at least $\frac{1}{v+1}$ and height
at least $\frac{1}{h+1}$. So we have a hole with aspect
ratio $\frac{v+1}{h+1}$ scaled so the smaller side has 
length $\frac{1}{v+1}$. 
The maximum scale factor of a rectangle with aspect ratio $p$ that fits
inside this hole is given by
$\frac{1}{v+1}\mathcal{C}_{\frac{v+1}{h+1}}(p)$.
We compare $\frac{1}{v+1}\mathcal{C}_{\frac{v+1}{h+1}}(p)$
to $D_k(p)$. 
Both $\frac{1}{v+1}\mathcal{C}_{\frac{v+1}{h+1}}(p)$
and $D_k(p)$ are piecewise functions,
we directly examine all $p\geq 1$ to show
$D_k(p)<\frac{1}{v+1}\mathcal{C}_{\frac{v+1}{h+1}}(p).$ 

For $1\leq p\leq \frac{(k+1)\lfloor\frac{k}{2}\rfloor}{\lceil\frac{k}{2} \rceil^2}$,
we have 
$D_k(p)=\frac{1}{k+1}$ and $\frac{1}{v+1}\mathcal{C}_{\frac{v+1}{h+1}}(p)=\frac{1}{v+1}.$ 
since $v<k$, we have $\frac{1}{k+1}<\frac{1}{v+1}$.
Then, for $\frac{(k+1)\lfloor\frac{k}{2}\rfloor}{\lceil\frac{k}{2} \rceil^2}\leq p\leq \frac{v+1}{h+1}$,  $D_k(p)$ 
decreases and $\frac{1}{v+1}\mathcal{C}_{\frac{v+1}{h+1}}(p)=\frac{1}{k+1}$ is constant.

For $\frac{v+1}{h+1}<p\leq w_{\frac{v+1}{h+1}}$, 
$D_k(p)=\min\Big\{  \frac{1}{k+1}\mathcal{C}_{k+1}(p), \frac{1}{\lceil\frac{k}{2}\rceil} \mathcal{C}_{\frac{\lceil\frac{k}{2}\rceil}{\lfloor\frac{k}{2}\rfloor}}(p)\Big\}
\leq
\frac{1}{\lceil\frac{k}{2}\rceil}\left(\frac{1}{p}\right)$
and

$\frac{1}{v+1}\mathcal{C}_{\frac{v+1}{h+1}}(p)=\frac{1}{v+1}\left(\frac{v+1}{h+1}\right)\frac{1}{p}.$
Since $h\leq \lfloor\frac{k}{2}\rfloor$,
we have 
$$\left(\frac{1}{\lceil\frac{k}{2}\rceil}\right)\left(\frac{1}{p}\right)<\frac{1}{h+1}\frac{1}{p}=\frac{1}{v+1}\left(\frac{v+1}{h+1}\right)\frac{1}{p}.$$

For $p\geq w_{\frac{v+1}{h+1}}$, we again solve the same constrained
optimization problem as in the even case. The minimum occurs when $v=h$.
This is a global minimum, so in the odd case the
minimum occurs when me make $v$ as close to $h$ as possible.
So, when $h< \lfloor\frac{k}{2}\rfloor<\lceil \frac{k}{2} \rceil< v$ 
we can fit the diagonal rectangle of 
$\mathcal{N}(\lceil\frac{k}{2}\rceil,\lfloor\frac{k}{2}\rfloor)$ inside the diagonal of the 
rectangle with dimensions $\frac{1}{v+1}\times\frac{1}{h+1}$
and the optimal scale factor must be larger.

\end{proof}

\section{Towards Axis-Aligned Net Optimality}
\label{append:append-axis}
In this appendix, we show that, for square
intruder, evenly spaced axis-aligned vertical lines are generally a local optimum.
This is a first step toward our conjecture that axis-aligned lines are
a global optimum for rectangular intruder.

\begin{lemma}\label{lem:local}
    Let $P$ be a intruder with aspect ration $p=1$, i.e., a square. Then evenly
    spaced axis-aligned vertical lines are a local optimum when the number of
    lines is $k >2$.
\end{lemma}

\begin{proof} 
    Denote the $k$ lines from left to right by $\ell_1, \ell_2, \ldots, \ell_k$.
    Assume, towards a contradiction, that, given some small $\epsilon > 0$,
    there are small shift and pivot values for each line that describe the
    translation of lines to an arrangement that results in a lower overall
    scaling factor $c$.  Specifically, let these values be denoted $s_1, s_2,
    \ldots, s_k$ and $p_1, p_2, \ldots , p_k$, respectively, where  $s_i, p_i
    \in [0, \epsilon]$, and at least one of these values $s_i$ or $p_i$ are
    nonzero.  Notice that, regardless of the pivot value, shifting any line
    decreases the maximum intruder size for one neighboring face, but increases it
    for the other neighboring face, leading to a higher $c$-value overall. Thus,
    we must have $s_1 = s_2 = \ldots = s_k$ = 0. This means we must consider the
    effect of pivot values on unshifted lines. Notice that pivoting a single
    line that is adjacent to a vertical (un-pivoted) line will increase the
    maximum intruder size for the face between them. Since the lines $p_1$ and $p_k$
    are always adjacent to the vertical edges of the bounding square, we must
    have $p_1 = p_k = 0$, otherwise the leftmost and rightmost faces would cause
    the arrangement to have a higher $c$-value. But then we must also have $p_2
    = p_{k-1} = 0$, or else the second to leftmost and second to rightmost faces
    would cause the arrangement to have a higher $c$-value. Continuing this line
    of argument, we eventually see that $p_1  =p_2 = \ldots = p_k = 0$,
    contradicting our assumption that some $s_i$ or $p_i$ be nonzero.
\end{proof}

\end{document}